\documentclass[11pt]{article}

\usepackage[utf8]{inputenc}
\usepackage[T1]{fontenc}
\usepackage{amsmath, amssymb, amsthm}
\usepackage{graphicx}
\usepackage{geometry}
\usepackage{hyperref}
\usepackage{enumerate}
\usepackage{enumitem}
\usepackage{algorithm}
\usepackage{algorithmic}
\usepackage{microtype} 
\usepackage{xspace} 
\usepackage{xcolor} 
\usepackage{cleveref}

\usepackage{graphicx}
\usepackage{authblk}

\usepackage{tikz}
\usetikzlibrary{arrows.meta, positioning, shapes.symbols, shadows}

\usepackage[T1]{fontenc}
\usepackage{fontawesome5}
\usepackage{array}

\usepackage[utf8]{inputenc}
\usepackage[T1]{fontenc}
\usepackage{amsmath, amssymb, amsthm}
\newtheorem{theorem}{Theorem}[section]
\newtheorem{lemma}[theorem]{Lemma}

\newtheorem{proposition}[theorem]{Proposition}

\newtheorem{example}[theorem]{Example}

\newcommand{\D}{\mathcal{D}}

\title{Sponsored Questions and How to Auction Them}

\author[1]{Kshipra Bhawalkar\thanks{\texttt{kshipra@google.com}}}
\author[2]{Alexandros Psomas\thanks{\texttt{apsomas@purdue.edu}}}
\author[1]{Di Wang\thanks{\texttt{wadi@google.com}}}

\affil[1]{Google Research}
\affil[2]{Purdue University and Google Research}

\date{}

\begin{document}

\maketitle

\begin{abstract}
 Online platforms connect users with relevant products and services using ads. A key challenge is that a user's search query often leaves their true intent ambiguous. Typically, platforms passively predict relevance based on available signals and in some cases offer query refinements.  The shift from traditional search to conversational AI provides a new approach. When a user's query is ambiguous, a Large Language Model (LLM) can proactively offer several clarifying follow-up prompts. 
 In this paper we consider the following: what if some of these follow-up prompts can be ``sponsored,'' i.e., selected for their advertising potential. 
 How should these ``suggestion slots'' be allocated? And, how does this new mechanism interact with the traditional ad auction that might follow?

This paper introduces a formal model for designing and analyzing these interactive platforms. We use this model to investigate a critical engineering choice: whether it is better to build an end-to-end pipeline that jointly optimizes the user interaction and the final ad auction, or to decouple them into separate mechanisms for the suggestion slots and another for the subsequent ad slot. We show that the VCG mechanism can be adopted to jointly optimize the sponsored suggestion and the ads that follow; while this mechanism is more complex, it achieves outcomes that are efficient and truthful. On the other hand, we prove that the simple-to-implement modular approach suffers from strategic inefficiency: its Price of Anarchy is unbounded.
\end{abstract}

\section{Introduction}

The rise of generative AI is fundamentally changing how users find information. Instead of a list of links, knowledge panels, etc., users now receive direct answers from Large Language Models (LLMs), followed by suggestions designed to guide their next step. This conversational interface is a powerful tool for resolving a major challenge in ad auctions: ambiguous user intent. For example, a user searching for ``running shoes'' might be looking for advanced trail running shoes or beginner road running shoes. While there are existing paradigms aiming to address this ambiguity in contemporary search, e.g., related search functionalities~\cite{google_adsense_related_search}, suggested chip refinements~\cite{GoogleSearchEEA}, etc., the shift to direct answers within a conversational flow makes resolving this intent even more pressing.  An LLM that offers distinct follow-up prompts or suggestions (``I would like more information about trail running shoes'') or direct questions (``Are you a beginner runner?''), can help clarify the user's needs in a single click. This not only improves the user experience by avoiding an overwhelming list of irrelevant information, but also creates a valuable new opportunity for advertisers: by sponsoring a suggestion like ``I want to know more about trail running shoes,'' an advertiser can guide a user with relevant intent directly toward their products, and produce better outcomes for the users, advertisers and platform.\footnote{Allowing advertiser interest to steer user conversation should be done with utmost care to ensure high user quality. Companies' policies and regulations might limit what is practically feasible. Our paper should be viewed as theoretically understanding what such a design would look like. We don't make any claims about feasibility from a policy or regulatory perspective, or about concrete plans of our employers.}

Given the limited display space and the varying utility advertisers derive from different suggestions, the platform faces a critical question: How should these ``suggestion'' slots be allocated? And, how does this new mechanism for allocating suggestions interact with the traditional ad auction that might follow? 
In this paper, we introduce a formal model for designing and analyzing these problems. We use this model to investigate a key design choice: whether to build an end-to-end mechanism that jointly optimizes the suggestion and downstream ad auction, or to adopt a simpler modular approach that decouples them into separate auctions.

\tikzset{
    box/.style={
        draw,
        fill=white,
        rounded corners=6pt, % Slightly smaller corners
        text width=6cm,      % Reduced width to fit in a column
        minimum height=0.8cm, % REDUCED HEIGHT
        inner xsep=8pt,
        align=center,
        font=\bfseries\small,
        line width=0.8pt,
        node distance=0.8cm % REDUCED DISTANCE
    },
    query/.style={
        box,
        fill=blue!15,        % Muted light blue
        draw=blue!50!black,  % Darker blue border
        text=black
    },
    system/.style={
        box,
        fill=gray!10,        % Very light gray
        draw=gray!50!black,
        text=black
    },
    action/.style={
        box,
        fill=orange!15,      % Muted light orange
        draw=orange!50!black,
        text=black
    },
    result/.style={
        box,
        fill=yellow!10,      % Light cream/yellow
        draw=orange!30!black,
        text=black,
        font=\normalfont\small
    },
    titlebox/.style={
        draw=black,
        line width=1.5pt,
        fill=gray!25,
        text width=6.5cm,    % Fits within the column width
        minimum height=0.8cm,
        align=center,
        font=\large\bfseries % Reduced font size
    },
    flowarrow/.style={
        ->,
        >=Stealth,
        thick,
        draw=gray!50!black,
        line width=1pt
    }
}

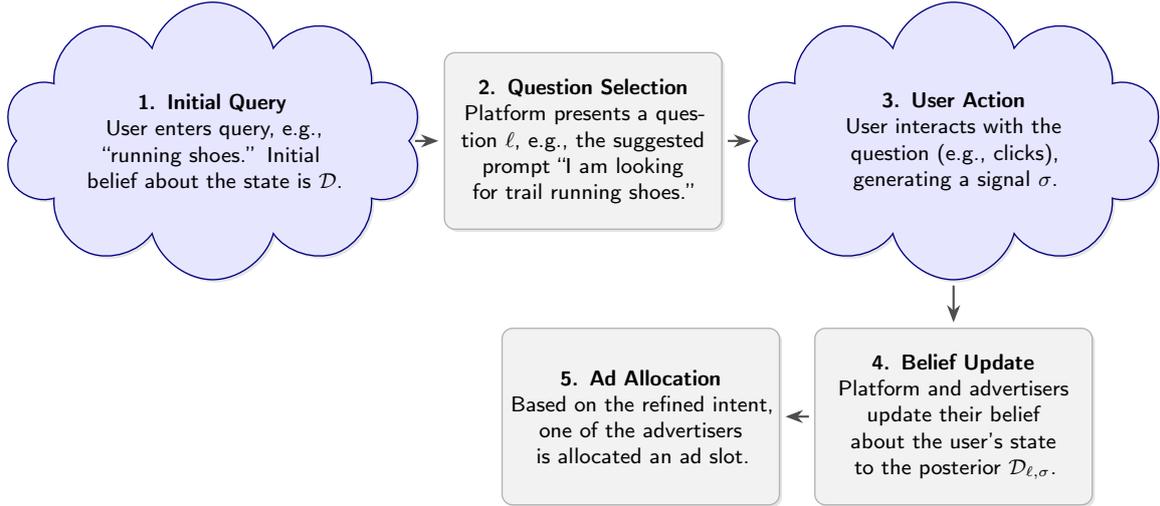
\begin{figure*}[t]
\centering
% The 'scale' option shrinks the entire figure proportionally.
% Adjust the value (e.g., 0.85, 0.95) as needed.
\begin{tikzpicture}[scale=0.9, transform shape,
    node distance=0.7cm and 0.5cm, % y-distance and x-distance
    % --- STYLES ---
    base/.style={
        draw,
        text width=3.8cm, % Slightly reduced width
        align=center,
        font=\sffamily\footnotesize, % Smaller font size
        minimum height=2.6cm, % Slightly reduced height
        line width=0.5pt,
        drop shadow={opacity=0.2, shadow xshift=1pt, shadow yshift=-1pt}
    },
    userInput/.style={
        base,
        cloud,
        cloud puffs=10,
        aspect=2,
        fill=blue!10,
        draw=blue!50!black
    },
    process/.style={
        base,
        rectangle,
        rounded corners=4pt,
        fill=gray!10,
        draw=gray!60
    },
    flow/.style={
        ->,
        >=Stealth,
        thick,
        draw=black!70,
        shorten >=2pt,
        shorten <=2pt,
        rounded corners=5pt
    }
]

% --- NODES (placed in a 2x3 grid) ---
% Top row
\node (step1) [userInput] {\textbf{1. Initial Query} \\ User enters query, e.g., ``running shoes.'' Initial belief about the state is $\mathcal{D}$.};
\node (step2) [process, right=of step1] {\textbf{2. Question Selection} \\ Platform presents a question $\ell$, e.g., the suggested prompt ``I am looking for trail running shoes.''};
\node (step3) [userInput, right=of step2] {\textbf{3. User Action} \\ User interacts with the question (e.g., clicks), generating a signal $\sigma$.};

% Bottom row (placed relative to the top row)
\node (step4) [process, below=of step3] {\textbf{4. Belief Update} \\ Platform and advertisers update their belief about the user's state to the posterior $\mathcal{D}_{\ell, \sigma}$.};
\node (step5) [process, left=of step4] {\textbf{5. Ad Allocation} \\ Based on the refined intent, one of the advertisers is allocated an ad slot.};

% --- ARROWS (connecting the nodes in a loop) ---
\draw [flow] (step1) -- (step2);
\draw [flow] (step2) -- (step3);
% The arrow going down
\draw [flow] (step3.south) -- (step4.north);

% Arrows going left on the bottom row
\draw [flow] (step4) -- (step5);

\end{tikzpicture}
\caption{A flowchart visualizing our model.}
\label{fig:example_flow_5step_small}
\end{figure*}

\subsection{Our contributions}

We introduce a new model for thinking about advertising in interactive platforms. Our model combines information elicitation and (ad) allocation. While directly inspired by modern applications like sponsored suggestions in conversational AI, a key strength of our model is its generality. We give a framework flexible enough to model many different forms of platform-user interaction, including sponsored clickable suggestions (prompts), sponsored questions (that the LLM makes), or other formats.

Our setup is as follows. A platform has a single ad slot for sale to one of $n$ strategic advertisers. The user's intent is captured by a state of nature $\theta$, that encapsulates complex, noisy user features and intent vectors that traditional keyword matching often misses. For the advertiser, $\theta$ represents the latent information about the user that is critical for determining an ad's relevance.
We assume $\theta$ is unknown to all parties and is sampled from a public prior distribution $\D$. Each advertiser $i$ has (1) a private \emph{base} value $v_i \geq 0$, the same across all states, and (2) a \emph{state-dependent} conversion rate $\alpha_i(\theta) \in \mathbb{R}_+$ that represents the ad's relevance for advertiser $i$ given $\theta$.\footnote{This can be, e.g., a click-though-rate, or a conversion value.} Advertisers have quasilinear utility, i.e., $u_i = v_i \cdot \alpha_i(\theta) - p_i$, where $p_i$ is the payment charged by the platform.

At the heart of our model is an \emph{information elicitation stage}: alongside any organic content presented to the user, the platform can choose to ask one of $k$ possible sponsored questions. Asking question $\ell$ generates a signal $\sigma$, observable by everyone, drawn from a distribution  $Q_{\ell}(. \mid \theta)$. Here a ``question'' is the LLM-generated prompt (``I am looking for trail running shoes''), and a ``signal'' is the publicly observable, measurable outcome of the user clicking or ignoring it. The probability of a user's response depends on their latent intent, a relationship captured by the probability distribution $Q_{\ell}(. \mid \theta)$. Given a signal $\sigma$ generated by question $\ell$, advertisers update their belief about the state $\theta$ to a common posterior distribution $\D_{\ell, \sigma}$. This, in turn, implies updated (expected) conversion rates for the ad slot.

\begin{example}[Running shoes]\label{ex:running-shoes}
A user is asking for information about ``running shoes.'' The platform is uncertain about their specific needs. We model the user's intent as the state of nature $\theta = (\theta_{\text{terrain}}, \theta_{\text{exp}}) \in \{0, 1\}^2$, where $\theta_{\text{terrain}} = 1$ if the user is a \emph{trail} runner and $0$ if they are a \emph{road} runner; $\theta_{\text{exp}} = 1$ if the user is an \emph{experienced} runner, and $0$ if they are a \emph{beginner}. The prior belief $\mathcal{D}$ is that all four states in $\Theta = \{ (0,0), (0,1), (1,0), (1,1) \}$ are equally likely; that is, $\mathcal{D}(\theta) = 1/4$ for all $\theta$.

There are two advertisers. \emph{Advertiser 1} sells professional trail running shoes. Their private base value is high, $v_1 = \$50$. Their conversion rate $\alpha_1(\theta)$ is highest for an experienced trail runner, e.g., $\alpha_1(1,1)=0.9$. It is lower for a beginner trail runner, $\alpha_1(1,0)=0.3$, and zero for all road runners ($\alpha_1(0, \cdot) = 0$). \emph{Advertiser 2} sells popular road running shoes for beginners. Their base value is $v_2 = \$30$. Their conversion rate $\alpha_2(\theta)$ is highest for a beginner road runner, $\alpha_2(0,0)=0.8$. It is lower for an experienced road runner, $\alpha_2(0,1)=0.4$, and zero for all trail runners ($\alpha_2(1, \cdot) = 0$).

Before the ad auction, the platform can show one of two possible sponsored suggestions. One possible choice is $\ell_{\text{trail}}$: ``I am looking for trail running shoes.'' If the user clicks on this suggestion, this generates a signal $\sigma_{\text{trail},\text{click}}$, which, for simplicity, we assume perfectly reveals their interest in trails; that is, if the signal is $\sigma_{\text{trail},\text{click}}$,  everyone learns that $\theta_{\text{terrain}}=1$.
The posterior distribution $\D_{\ell_{\text{trail}},\sigma_{\text{trail},\text{click}}}$ is the uniform distribution over $(1,0)$ (trail, beginner) and $(1,1)$ (trail, experienced).
If the user does not click on this suggestion this generates a signal $\sigma_{\text{trail},\text{no click}}$, which, for simplicity, we assume perfectly reveals the user's lack of interest in trails. $\D_{\ell_{\text{trail}},\sigma_{\text{trail},\text{no click}}}$ is the uniform distribution over $(0,0)$ (road, beginner) and $(0,1)$ (road, experienced).
A different choice for a suggestion is $\ell_{\text{targeted}}$: ``I am an experienced trail runner.'' If this is clicked, we assume it completely reveals that the state is $\theta = (1,1)$; otherwise, it is revealed that the state is not $\theta = (1,1)$, yielding a posterior distribution that is uniform over $(0,0)$ (road, beginner), $(1,0)$ (trail, beginner), and $(0,1)$ (road, experienced).
\end{example}

In this paper, we are interested in understanding the trade-offs between two fundamental design philosophies: an end-to-end approach that jointly optimizes the suggestion and auction stages, versus a modular approach that decouples them.

We first analyze the end-to-end approach, where advertisers report their base values, and the platform uses this information to make decisions on which question to present, then which ad to show, and how much to charge the advertisers. We show that the platform can instantiate this as an end-to-end VCG auction. This auction is truthful, i.e., incentivizes the advertisers to report their true base values. An interesting observation is that the VCG payment can be decomposed into a second price payment (coming from the ad allocation step) and an additional payment based on the influence an advertiser's bid had in the selection of the question. While a unified VCG auction may seem complex, we argue it is a strong candidate for practical systems, especially those with autobidders, because of its efficiency and truthfulness guarantees --- guarantees that, as we will show, are lacking in simpler modular approaches.

Next, we consider a natural alternative to the end-to-end approach: a modular mechanism that decouples the auction into two stages (one for the question, and one for the ad). We consider a VCG-per-stage design: advertisers first bid on which question to ask, and after a question is chosen and a signal is observed, they bid again for the final ad slot.
In this setup, an advertiser's value for a given question is simply their expected utility from the subsequent ad auction. We prove that this intuitive behavior constitutes a pure Nash equilibrium: advertisers bid this expected utility in the first stage and their true value in the second. 

This seemingly elegant equilibrium hides significant practical challenges. For example, an advertiser's value for a question is difficult to compute, as it depends on the unknown values and strategies of all other bidders. A tempting solution is for the platform to assist. Perhaps, for example, it could collect advertisers' base values and compute the stage-1 bids on their behalf. We show that this ``proxy'' approach creates incentives for advertisers to misreport their underlying values to manipulate the outcome. Furthermore, the proposed bidding language (bidding on every question) fails to capture the context-dependent nature of conversational AI, where an advertiser's value for a question can change dramatically based on the preceding dialogue. For example, an advertiser might value a question about trail running differently if they know that the user is an experienced runner; asking them to bid on the question by itself might miss such nuances.

Even if these practical hurdles could be overcome, the modular approach suffers from a fatal flaw. We prove that the equilibrium of the VCG-per-stage mechanism has an \emph{unbounded} Price of Anarchy. That is, its outcome can be arbitrarily less efficient than the one achieved by the unified, end-to-end auction. This severe inefficiency is not an artifact of the VCG-per-stage design; we find that similar issues arise when using other common formats, such as first-price or all-pay auctions for the initial question slots.

The main take-home message is that a unified, end-to-end mechanism is a more attractive and robust solution for allocating sponsored suggestions. If a modular design is pursued for its apparent simplicity, our results show that significant further research is needed to develop a mechanism that is both strategically simple for advertisers and achieves good social welfare.

In summary, our contributions are as follows:
\begin{itemize}
    \item We introduce a formal model for sponsored suggestions in interactive platforms.
    \item We observe that the end-to-end VCG mechanism is truthful (DSIC) in this setting and provide a simple characterization of its payments.
    \item We identify and characterize an intuitive pure Nash equilibrium for the natural ``VCG-per-stage'' modular mechanism.
    \item We prove that this modular mechanism has an unbounded Price of Anarchy, and show this severe inefficiency persists even when the platform uses first-price or all-pay auctions in the suggestion stage.
\end{itemize}

\subsection{Related Work}
Many search engines and other products show query refinements or suggested queries to help users. For most popular search engines, there are policies which commit to showing the best results for the user, and which preclude them from considering commercial interest (such as advertiser bids or ad revenue) to optimize these queries~\cite{Microsoft_Bing_Results,Google_Search_Approach}. Exceptions exist; for example, Google's Relevant Search for Content (RSOC), allows publishers to use ``funnel RPM'' in optimization of suggested queries shown on their webpage~\cite{google_adsense_related_search}, thus allowing advertisers to inform which suggestions are shown to the user (note though that Google AFS policies~\cite{Google_AFS_Policies} still require that first and foremost user experience is optimized).

Closer to our interest here, there is growing interest from practice to integrate advertising into conversational search. For instance, Perplexity recently launched an experiment featuring sponsored follow-up prompts~\cite{PerplexityAds}. 
At the same time, the academic literature on the topic is nascent but rapidly growing.  Feizi et al.~\cite{feizi2025online} discuss challenges and opportunities in this area.
One major stream of research focuses on how to natively integrate ads into a final, LLM-generated answer.
Various models have been proposed, including token-based auctions~\cite{duetting2024mechanismdesignlargelanguage}, bidding for placement within an LLM summary~\cite{dubey2024}, fusing organic and sponsored content~\cite{mordo2024}, embedding ads in an independently generated organic answer~\cite{hajiaghayi2024ad}, and bidding to influence the fine-tuning of the LLM itself~\cite{soumalias2025truthfulaggregationllmsapplication}. 
Banchio et al.~\cite{banchio2025ads} study the problem of ad placement in a conversational search; specifically, they study the equilibria under first- and second-price auctions with respect to the timing of showing an ad. In contrast, our work models the preceding, interactive stage as well: using suggestions to influence the \emph{flow} of the conversation itself.
Bergemann et al.~\cite{BergemannEC25} study a setting where agents have private information about both their preferences and a common state, showing that standard mechanisms fail. They propose a solution via ``data-driven mechanisms,'' which use post-allocation data (like observed clicks) to adjust payments and restore truthfulness. While our information structure differs (advertisers' state-dependent conversion rates are public information in our paper), our work addresses the same high-level challenge for the platform: designing a mechanism that must jointly elicit preferences and handle uncertainty about a common state.
\section{Preliminaries}

We consider the problem of a platform with a single ad slot for sale, to one of $n$ strategic advertisers. We will refer to the ad slot as the \emph{item}.
The item is described by a state of nature $\theta \in \Theta$, for some finite set $\Theta$, that is initially unknown to both the platform and the advertisers. $\theta$ is sampled from a known distribution $\D$; we write $\D(\theta)$ for the probability that the state of nature is $\theta$.

Each advertiser $i$ has a private \emph{base} value $v_i \geq 0$ for the item, which stays the same across all states $\theta$. There is also a publicly known \emph{state-dependent} conversion rate $\alpha_i(\theta) \geq 0$ that represents the item's relevance for advertiser $i$ given $\theta$.\footnote{This can be, e.g., a click-through rate, or a conversion value.} The value of advertiser $i$ for an item with state $\theta$ is the base value multiplied by the conversion rate, i.e., $v_i \cdot \alpha_i(\theta)$. Advertisers have quasilinear utilities: when allocated an item with state $\theta$ for a payment $p_i$, advertiser $i$ has utility $u_i = v_i \cdot \alpha_i(\theta) - p_i$.

Before the item is allocated, the platform presents one of $k$ possible questions. Presenting question $\ell$ reveals a random signal $\sigma \in \Sigma_{\ell}$; $\sigma$ is drawn from a conditional distribution over signals $Q_{\ell}(. \mid \theta)$. 
The marginal signal distribution under question $\ell$ is
\[
S_{\ell}(\sigma)\ :=\ \sum_{\theta\in\Theta} Q_{\ell}(\sigma \mid \theta)\,\D(\theta),\qquad \sigma\in\Sigma_\ell.
\]

Upon observing $\sigma$, the posterior over states is
\[
\D_{\ell,\sigma}(\theta)\ := \frac{\D(\theta)\,Q_{\ell}(\sigma \mid \theta)}{S_{\ell}(\sigma)},\qquad \theta \in \Theta, S_\ell(\sigma)>0.
\]
Therefore, upon observing signal $\sigma$ from question $\ell$, the posterior expected conversion rate for advertiser $i$ is
\[
\alpha_i(\ell,\sigma)\ :=\ \sum_{\theta\in\Theta} \alpha_i(\theta)\,\D_{\ell,\sigma}(\theta).
\]

We call $v_i\,\alpha_i(\ell,\sigma)$ the \emph{effective value} of advertiser $i$ given $(\ell,\sigma)$.

\begin{example}[Running shoes (continuation of \Cref{ex:running-shoes})]\label{ex:running-shoes-formal}
Recall that $\Theta=\{(\theta_{\text{terr}},\theta_{\text{exp}})\in\{0,1\}^2\}$ with uniform prior $\D$. We have two advertisers with private base values $v_1=50$, $v_2=30$, and public conversion rates:
\[
\alpha_1(\theta)=
\begin{cases}
0.9 & \text{if }(\theta_{\text{terr}},\theta_{\text{exp}})=(1,1),\\
0.3 & \text{if }(\theta_{\text{terr}},\theta_{\text{exp}})=(1,0),\\
0 & \text{if }\theta_{\text{terr}}=0,
\end{cases}
\]
\[
\alpha_2(\theta)=
\begin{cases}
0.8 & \text{if }(\theta_{\text{terr}},\theta_{\text{exp}})=(0,0),\\
0.4 & \text{if }(\theta_{\text{terr}},\theta_{\text{exp}})=(0,1),\\
0 & \text{if }\theta_{\text{terr}}=1.
\end{cases}
\]

There are two questions:
\begin{enumerate}[leftmargin=*]
    \item \textbf{Terrain question} $\ell_{\text{terr}}$ (``I am looking for trail running shoes?'') with signals $\Sigma_{\text{terr}}=\{\text{click},\text{no-click}\}$ and
            \begin{gather*}
            Q_{\text{terr}}(\text{click}\mid\theta)=\mathbf{1}\{\theta_{\text{terr}}=1\}. \\
            Q_{\text{terr}}(\text{no-click}\mid\theta)=\mathbf{1}\{\theta_{\text{terr}}=0\}.
            \end{gather*}
        Thus, $S_{\text{terr}}(\text{click})=S_{\text{terr}}(\text{no-click})=\tfrac12$.
        The posteriors are: $\D_{\text{terr},\text{click}}$ is uniform over $\{(1,1),(1,0)\}$, and
        $\D_{\text{terr},\text{no-click}}$ is uniform over $\{(0,1),(0,0)\}$.
        Hence, the posterior expected conversion rates are
                \begin{align*}
                \alpha_1(\text{terr},\text{click})   &= \tfrac{0.9+0.3}{2}=0.6, &
                \alpha_1(\text{terr},\text{no-click})&= 0, \\
                \alpha_2(\text{terr},\text{click})   &= 0,                      &
                \alpha_2(\text{terr},\text{no-click})&= \tfrac{0.4+0.8}{2}=0.6.
                \end{align*}
        Therefore, the effective value of  advertiser 1 is $50\cdot 0.6=30$ when the signal is ``click,'' and $0$ otherwise. The effective value of  advertiser 2 is $0$ when the signal is ``click,'' and $30\cdot 0.6=18$ otherwise. Therefore, ex-ante, $\ell_{\text{terr}}$ induces a lottery that with probability $1/2$ gives effective values $30$ and $0$, and with probability $1/2$ the effective values are $0$ and $18$.
    \item \textbf{Targeted question} $\ell_{\text{tgt}}$ (``I am an experienced trail runner'') with signals $\Sigma_{\text{tgt}}=\{\text{click},\text{no-click}\}$ and
            \begin{gather*}
            Q_{\text{tgt}}(\text{click}\mid\theta)=\mathbf{1}\{\theta=(1,1)\}. \\
            Q_{\text{tgt}}(\text{no-click}\mid\theta)=1-Q_{\text{tgt}}(\text{click}\mid\theta).
            \end{gather*}
            Thus $S_{\text{tgt}}(\text{click})=\tfrac14$ and $S_{\text{tgt}}(\text{no-click})=\tfrac34$. The posteriors are: $\D_{\text{tgt},\text{click}}$ is a point mass on $(1,1)$, and
            $\D_{\text{tgt},\text{no-click}}$ is uniform over $\{(1,0),(0,1),(0,0)\}$.
            Hence the posterior expected conversion rates are
            \[
            \begin{aligned}
            \alpha_1(\text{tgt},\text{click})   &= 0.9, &
            \alpha_1(\text{tgt},\text{no-click})&= \tfrac{0.3+0+0}{3}=0.1,\\
            \alpha_2(\text{tgt},\text{click})   &= 0,   &
            \alpha_2(\text{tgt},\text{no-click})&= \tfrac{0+0.4+0.8}{3}=\tfrac{1.2}{3}=0.4.
            \end{aligned}
            \]
            Therefore, the effective value of advertiser 1 is $50\cdot 0.9=45$ when the signal is ``click,'' and $50 \cdot 0.1 = 5$ otherwise. The effective value of advertiser 2 is $0$ when the signal is ``click,'' and $30\cdot 0.4=12$ otherwise. Therefore, ex-ante, $\ell_{\text{tgt}}$ induces a lottery that with probability $1/4$ gives effective values $45$ and $0$, and with probability $3/4$ the effective values are $5$ and $12$.
\end{enumerate}

\end{example}

After the question is presented, and the signal is randomly drawn and revealed to all parties, the platform can allocate the item.

Overall, an instance of our problem is parameterized by: (1) the public prior distribution $\D$ over states, (2) public conditional distributions $Q_{\ell}(\sigma \mid \theta)$ for every question $\ell$, signal $\sigma$, and state $\theta$,
(3) private base values $v_i$ for each advertiser $i$, and (4) public conversion rates $\alpha_i(\theta)$ for each advertiser $i$ and state $\theta$.

Given an instance, the timing of events is as follows:
\begin{enumerate}
    \item Nature samples the item's state $\theta$ from $\D$. $\theta$ is not revealed to the platform or the advertisers.
    \item The advertisers possibly interact with the platform.
    \item The platform presents a question $\ell$.
    \item Nature draws a signal $\sigma$ from $Q_{\ell}(. \mid \theta)$, and reveals it to the platform and the advertisers.
    \item The advertisers possibly interact with the platform.
    \item The platform allocates the item to one of the advertisers (possibly in a randomized way), and charges payments.
\end{enumerate}

Given this timing, the platform picks the rules of the game, and commits to them ex ante. The rules may condition on publicly observed signals $\sigma$ and any messages solicited from advertisers in steps (2) and (5). In particular, the platform specifies (i) what input (if any) advertisers provide in steps (2) and (5); (ii) how the question $\ell$ is selected in step (3); and (iii) the allocation and payment rules used in step (6). The platform's goal is to maximize expected welfare: the expected sum, over the platform's (possibly randomized) choice of question and the induced signal, of each advertiser's base value times their (posterior) conversion rate. The goal of each advertiser is to maximize expected utility: its base value times conversion rate minus its payment.

In this paper, we study two different approaches for solving the platform problem. First, in the \emph{direct revelation} approach, the platform asks each advertiser to report their base values in step (2) (no interaction occurs in step (5)) and commits to a truthful mechanism. We define the setting in detail, and study such mechanisms, in \Cref{sec: direct}. Second, we consider \emph{modular two-stage mechanisms}: the platform commits to two auctions: one for picking the question in step (3), and one for allocating the item in step (6). Messages/bids are therefore solicited in both steps (2) and (5). Here, performance is evaluated at a Nash equilibrium. We define the setting in detail, and study such mechanisms in~\Cref{sec: modular}.

\section{Direct Revelation Mechanisms}\label{sec: direct}

In this section, we study the design of a single, end-to-end mechanism for both information acquisition and allocation. This corresponds to the direct revelation approach, where the platform commits to a set of rules, asks advertisers to report their private information (their base values), and then manages the entire process (which question to ask, who gets the item, how much to charge). The timing of a direct mechanism is as follows:

\begin{enumerate}[leftmargin=*]
    \item Nature samples the item's state $\theta$ from $\D$. $\theta$ is not revealed to the platform or the advertisers.
    \item Each advertiser $i$ submits to the platform a bid $b_i$ (possibly different than their base value $v_i$).
    \item The platform presents a question $\ell$.
    \item Nature draws a signal $\sigma$ from $Q_{\ell}(. \mid \theta)$, and reveals it to the platform and the advertisers.
    \item The platform allocates the item to one of the advertisers (possibly in a randomized way), and charges payments.
\end{enumerate}

A direct mechanism $\mathcal{M} = (\lambda, q , p)$ consists of three functions, that depend on the reported bids $b = (b_1, b_2, \dots, b_n)$, from the single interaction with the advertisers in step (2) above:
\begin{itemize}[leftmargin=*]
    \item The \emph{suggestion rule} $\lambda(b)$ is a probability distribution over the available questions; we write $\lambda_{\ell}(b)$ for the probability that the mechanism picks question $\ell$. It holds that $\sum_{\ell = 1}^k \lambda_{\ell}(b) = 1$, for all $b$.
    \item The \emph{allocation rule} $q(b, \sigma)$: For a given report $b$ and a realized signal $\sigma$ (drawn from $Q_{\ell}(. \mid \theta)$, where $\ell$ is the question sampled from $\lambda(b)$), $q_i(b, \sigma)$ is the probability that advertiser $i$ receives the item. It holds that $\sum_{i = 1}^n q_i(b, \sigma) \leq 1$, for all $b$ and $\sigma$. 
    \item The \emph{payment rule} $p(b, \sigma)$: For a given report $b$ and a realized signal $\sigma$, $p_i(b, \sigma)$ is the payment charged to advertiser $i$.
\end{itemize}

For deterministic mechanisms, $\lambda_{\ell}(b)$ and $q_i(b,\sigma)$ are indicators for $\ell$ being chosen/advertiser $i$ winning the item.

Let $x_i(b) = E_{\ell \sim \lambda(b), \sigma \sim S_\ell(.)}[ \alpha_i(\ell, \sigma) \cdot q_i(b, \sigma) ]$ be the expected \emph{delivered} conversion rate for advertiser $i$, where the expectation is the randomness of the mechanism and the random signal. Intuitively, if $\alpha_i(\ell, \sigma)$ is a click-through rate, $x_i(b)$ is the ``expected number of clicks'' allocated to advertiser $i$. Similarly, slightly overloading notation, let $p_i(b) = E_{\ell \sim \lambda(b), \sigma \sim S_\ell(.)}[ p(b, \sigma) ]$ be the expected payment of advertiser $i$.
The interim utility of advertiser $i$, with a true base value $v_i$, who submits a bid $b_i$, and others report $b_{-i}$ is therefore $U_i(v_i \rightarrow b_i; b_{-i}) = v_i \cdot x_i(b_i; b_{-i}) - p_i(b_i; b_{-i})$. 
Our notion of truthfulness is \emph{dominant strategy incentive compatibility}. A mechanism is dominant strategy incentive compatible (DSIC) if truthtelling is a dominant strategy: $U_i(v_i \rightarrow v_i; b_{-i}) \geq U_i(v_i \rightarrow b_i; b_{-i})$ for all $v_i, b_i$ and $b_{-i}$. We are interested in designing a DSIC mechanism that maximizes expected social welfare.

\subsection{Results}

Our first result establishes that the welfare-maximizing direct mechanism is truthful. This is achieved by defining the suggestion and allocation rules to greedily maximize the reported welfare, and then applying the classic Vickrey-Clarke-Groves (VCG) payment structure.

\begin{lemma}\label{lem: vcg works}
The direct mechanism $\mathcal{M}^* = (\lambda^*, q^*, p^{VCG})$ that maximizes expected social welfare is Dominant Strategy Incentive Compatible (DSIC), where: 
\begin{itemize}[leftmargin=*]
    \item The suggestion rule $\lambda^*(b)$ deterministically chooses the question $\ell^* = \ell^*(b)$ that maximizes the ex-ante expected welfare, i.e., $\ell^* = \operatorname*{argmax}_{\ell} E_{\sigma \sim S_{\ell}(.) } [ max_{j=1,\dots,n} b_j \cdot \alpha_j(\ell, \sigma) ]$
    \item The allocation rule $q^*(b, \sigma)$ allocates the item to the advertiser $i^*$ with the highest \emph{posterior} value after suggestion $\ell^*(b)$ is presented, and signal $\sigma$ is observed $i^*(b, \sigma) = \operatorname*{argmax}_{j \in \{1,\dots,n\}} b_j \cdot \alpha_j(\ell^*(b), \sigma).$ That is, $q_{i^*} = 1$ and $q_j = 0$ for $j \neq i^*$.
    \item The payment rule $p^{VCG}$ charges each advertiser $i$ the externality they impose on all other advertisers:  the expected welfare of the optimal solution without $i$ in the instance, minus the expected welfare of everyone other than $i$ in the optimal solution with $i$.
\end{itemize}
\end{lemma}

Before proving the lemma, we give a simple example of how the mechanism works.

\begin{example}[Mechanism $\mathcal{M}^*$ applied to~\Cref{ex:running-shoes-formal}]
Consider applying $\mathcal{M}^*$ to the instance from~\Cref{ex:running-shoes-formal} given (truthful) bids $b_1=50$, $b_2=30$. Given these bids, the mechanism first chooses a question to present.

\paragraph{Choosing the question $\ell^*(b)$.}
As discussed in~\Cref{ex:running-shoes-formal}, question $\ell_{\text{terr}}$ induces a lottery that, with probability $1/2$, gives effective values (base value times posterior expected conversion rate) $30$ and $0$, and with probability $1/2$, the effective values are $0$ and $18$. Therefore, the expected welfare under $\ell_{\text{terr}}$ is $\frac{1}{2} \cdot 30 + \frac{1}{2} \cdot 18 = 24$.
On the other hand, $\ell_{\text{tgt}}$ induces a lottery that, with probability $1/4$, gives effective values $45$ and $0$, and with probability $3/4$, the effective values are $5$ and $12$. Therefore, the expected welfare under $\ell_{\text{tgt}}$ is $\frac{1}{4} \cdot 45 + \frac{3}{4} \cdot 12 = 20.25$.
Hence the mechanism picks $\ell^*(b) = \ell_{\mathrm{terr}}$.

\paragraph{Allocating the item.}
Since $\ell_{\text{terr}}$ is picked, then with probability $1/2$ the signal is ``click,'' the induced effective values are $30$ and $0$, and advertiser 1 receives the item. Otherwise, the signal is ``no-click,'' the induced values are $0$ and $18$, and advertiser 2 receives the item.

\paragraph{Payments.}
The expected value of advertiser 1 in the current solution is $15$: they receive the item when the signal from asking $\ell_{\text{terr}}$ is ``click'' (this is a probability $1/2$ event), and their effective value is $30$ when this happens. For advertiser 2, their expected value is $9$: they receive the item when the signal from asking $\ell_{\text{terr}}$ is ``no-click'' (this is a probability $1/2$ event), and their effective value is $18$ when this happens.

What would happen if advertiser 1 did not participate? In that case, the question would be picked to optimize just for advertiser 2, who is indifferent between the two questions, since its the expected value is (still) $9$ under both choices. This means that advertiser 1 does not impose any externality, so its payment is zero.

What would happen if advertiser 2 did not participate? Similarly, the question would be picked to optimize just for advertiser 1, who is also indifferent between the two questions: $\ell_{\text{terr}}$ gives an expected value of $30/2 = 15$, and $\ell_{\text{tgt}}$ gives an expected value of $45/4 + 5 \cdot 3/4 = 15$. Therefore that advertiser 2 does not impose any externality, so its payment is also zero.
\end{example}

We will now prove Lemma~\ref{lem: vcg works} which shows that mechanism $\mathcal{M}^*$ is DSIC.
\begin{proof}[Proof of~\Cref{lem: vcg works}]
Fix bids $b$. For any question $\ell$ and signal $\sigma$, the reported welfare from allocating to $i$ is $b_i\,\alpha_i(\ell,\sigma)$. Hence, the maximum reported welfare for fixed $\ell$ is $E_{\sigma\sim S_\ell}\!\Big[\max_j b_j\,\alpha_j(\ell,\sigma)\Big]$,
attained by the allocation rule $q^*(b, \sigma)$, by allocating, for each realized $\sigma$, to a maximizer of $b_j\,\alpha_j(\ell,\sigma)$. Optimizing over $\ell$, gives the suggestion rule $\lambda^*(b)$ above; therefore, our rules maximize expected reported welfare over all feasible outcomes (choice of $\ell$ and signal-contingent allocation). 
Define
\[
W_{-i}(b_{-i})\ :=\ \max_{\ell\in[k]} E_{\sigma\sim S_\ell}\!\Big[\max_{j\neq i} b_j\,\alpha_j(\ell,\sigma)\Big].
\]
Advertiser $i$'s payment is then:
\[
p_i^{VCG}(b)\ =\ W_{-i}(b_{-i}) - E_{\sigma\sim S_{\ell^*(b)}} \Big[\sum_{j\neq i} b_j \, \alpha_j\big(\ell^*(b),\sigma\big) \, q_j^*(b,\sigma)\Big].
\]
Therefore, $i$'s expected utility when its true base value is $v_i$ and it reports $b_i$ is:
$U_i(v_i \rightarrow b_i; b_{-i}) = E_{\sigma\sim S_{\ell^*(b)}}\Big[v_i \, \alpha_i\big(\ell^*(b),\sigma\big)\,q_i^*(b,\sigma) + \sum_{j\neq i} b_j\,\alpha_j\big(\ell^*(b),\sigma\big)\,q_j^*(b,\sigma)\Big]\ -\ W_{-i}(b_{-i})$.
Since $(\lambda^*,q^*)$ selects, for each $b$, the outcome that maximizes $E_{\sigma\sim S_{\ell^*(b)}}[\sum_j b_j \alpha_j\big(\ell^*(b),\sigma \big)]$, the first term in $U_i(v_i \rightarrow b_i; b_{-i})$ is maximized when $b_i=v_i$. The second term ($W_{-i}(b_{-i})$) does not depend on $i$'s report; therefore, $b_i = v_i$ is a dominant strategy.
\end{proof}

The next lemma gives a simplified way of calculating the VCG payments in our problem. Specifically, the VCG payment of advertiser $i$ can be decomposed into (1) an externality from stage 1, plus (2) an expected second price payment.

\begin{lemma}
Let $\mathrm{SW}_{-i}(\ell):=E_{\sigma\sim S_\ell}\!\big[\max_{j\neq i} b_j\,\alpha_j(\ell,\sigma)\big]$. Let $\ell^* = \ell^*(b)$ be the question chosen by $\mathcal{M}^* = (\lambda^*, q^*, p^{VCG})$. Then
\[
\begin{aligned}
p_i^{VCG}(b)
&=\underbrace{\big(\max_{\ell}\mathrm{SW}_{-i}(\ell)-\mathrm{SW}_{-i}(\ell^*)\big)}_{\text{Stage--1 externality}} +
\underbrace{E_{\sigma\sim S_{\ell^*}}\!\Big[\max_{j\neq i} b_j\,\alpha_j(\ell^*,\sigma)\,q_i^*(b,\sigma)\Big]}_{\text{expected Stage--2 second price payment}}.
\end{aligned}
\]
\end{lemma}

\begin{proof}
Starting from the definition of $p_i^{VCG}(b)$ we have
\[
p_i^{VCG}(b)
= \max_{\ell}\mathrm{SW}_{-i}(\ell)\;-\;
E_{\sigma\sim S_{\ell^*}}\!\Big[\sum_{j\neq i} b_j\,\alpha_j(\ell^*,\sigma)\,q_j^*(b,\sigma)\Big]. \tag{1}
\]

Let $s_i(\sigma):=\max_{j\neq i} b_j\alpha_j(\ell^*,\sigma)$.
For every realized $\sigma$, the welfare-maximizing $q^*(b,\sigma)$ allocates only to maximizers of
$b_j\alpha_j(\ell^*,\sigma)$. Therefore, for every signal $\sigma$
\[
\sum_{j\neq i} b_j\alpha_j(\ell^*,\sigma)\,q_j^*(b,\sigma)
= s_i(\sigma)\sum_{j\neq i} q_j^*(b,\sigma).
\]
If $s_i(\sigma)>0$, since $\sum_j q_j^*(b,\sigma)=1$, we have $\sum_{j\neq i} q_j^*(b,\sigma)=1-q_i^*(b,\sigma)$.
If $s_i(\sigma)=0$, both sides are $0$ regardless of $q^*$. Therefore, for all $\sigma$,
\[
\sum_{j\neq i} b_j\alpha_j(\ell^*,\sigma)\,q_j^*(b,\sigma)
= s_i(\sigma)\,\big(1-q_i^*(b,\sigma)\big).
\tag{2}
\]
Taking an expectation in (2) and substituting into (1) gives
\[
\begin{aligned}
&p_i^{VCG}(b)
=\max_{\ell}\mathrm{SW}_{-i}(\ell)
- E_{\sigma\sim S_{\ell^*}}[s_i(\sigma)]
+ E_{\sigma\sim S_{\ell^*}}[s_i(\sigma)\,q_i^*(b,\sigma)]\\
&\quad\quad=\underbrace{\big(\max_\ell \mathrm{SW}_{-i}(\ell)-\mathrm{SW}_{-i}(\ell^*)\big)}_{\text{Stage--1 externality}}
+
\underbrace{E_{\sigma\sim S_{\ell^*}}\!\big[s(\sigma)\,q_i^*(b,\sigma)\big]}_{\text{expected Stage--2 second price}}.
\end{aligned}
\]
\end{proof}
\section{Modular Mechanisms: Decoupling Questions and Allocation}\label{sec: modular}

In contrast to the end-to-end approach, in this section we study \emph{modular mechanisms} where the platform decouples the problem into two distinct stages. First, it runs an auction to decide which question to present. Second, after the signal from that question is publicly observed, it runs a separate auction to allocate the item. This approach is simpler to implement, but requires advertisers to act strategically.
The timing of a modular mechanism is as follows:

\begin{enumerate}[leftmargin=*]
    \item Nature samples the item's state $\theta$ from $\mathcal{D}$. $\theta$ is not revealed to the platform or the advertisers.
    \item Each advertiser $i$ submits a vector of bids $b_i^Q = (b_{i,1}^Q, \dots, b_{i,k}^Q)$, where $b_{i,\ell}^Q$ is their bid for having question $\ell$ presented.
    \item The platform selects a question $\ell$ and charges payments based on the bids $\{b_i^Q\}_{i=1}^n$.
    \item Nature draws a signal $\sigma$ from $Q_{\ell}(. | \theta)$, and reveals it to the platform and the advertisers.
    \item Each advertiser $i$ submits a bid $b_i^A$ for the item.
    \item The platform allocates the item and charges payments based on the bids $\{b_i^A\}_{i=1}^n$.
\end{enumerate}

Analytically, we will treat this as a \textbf{one-shot normal-form game}: before play begins, each advertiser commits to a contingent plan for both stages. That is, advertisers \emph{simultaneously} submit (i) a bid for each question and (ii) a bidding function for the item, mapping $(\ell, \sigma)$ to a stage-2 bid. The subsequent ``stages'' are merely the platform executing those plans after the signal is realized. Later in this section, we prove that a certain behavior is a Nash equilibrium under a certain modular mechanism in this normal form game; we also note, however, that the same behavior is a subgame-perfect equilibrium for the corresponding sequential game.

A modular mechanism $\mathcal{M} = (\mathcal{M}_Q, \mathcal{M}_A)$ is composed of two separate mechanisms. The stage 1 mechanism, $\mathcal{M}_Q$, takes as input bids $\{b_i^Q\}$, and consists of: (i) a \emph{question selection rule} $\lambda(b^Q)$, which gives a probability distribution over the available questions; we write $\lambda_{\ell}(b^Q)$ for the probability that the mechanism picks question $\ell$ and have that $\sum_{\ell} \lambda_{\ell}(b^Q) = 1$. (ii) a stage 1 payment rule $p^Q(b^Q)$ that charges a payment $p^Q_i(b^Q)$ to each advertiser $i$.
The stage 2 mechanism, $\mathcal{M}_A$, takes as input bids $\{b_i^A\}$, and consists of: (i) an \emph{item selection rule} $q(b^A; \ell, \sigma)$, that allocates the item to advertiser $i$ with probability $q_i(b^A; \ell, \sigma)$, and (ii) a stage 2 payment rule $p^A(b^A; \ell, \sigma)$ that charges each advertiser $i$ a payment $p_i^A(b^A; \ell, \sigma)$.

We evaluate modular mechanisms in their worst-case \emph{Nash} equilibria. A strategy $s_i$ for an advertiser $i$ specifies what they will do at every interaction point. That is, a strategy $s_i$ consists of (1) a Stage 1 bidding vector $b^Q_i = (b^Q_{i,1}, \dots, b^Q_{i,k})$, which specifies their bid for each question $\ell$, and (2) a bidding function $b^A_i(\ell, \sigma)$ which specifies what they will bid for the item for every possible question $\ell$ and signal $\sigma$ that could be revealed. Let $s = (s_1, \dots, s_n)$ be a strategy profile for the advertisers. This profile determines the bids in both stages, and thus the final outcome and advertisers' utilities. Let $U_i(s_i; s_{-i})$ be the expected utility for advertiser $i$ when they play strategy $s_i$ and all other advertisers play according to strategies in $s_{-i}$, where the expectation is over the strategies of others, the environment (random state of the item, signal), and randomness in the modular mechanism. A strategy profile $s^*$ is a Nash equilibrium if no advertiser has a profitable unilateral deviation; that is, if for every advertiser $i$, and any strategy $s_i$, $U_i(s^*_i; s^*_{-i}) \geq U_i(s_i; s^*_{-i})$.

We compare the social welfare of the Nash equilibrium to that of the optimal allocation assuming full information. Note that the optimal allocation is an example of what the direct revelation mechanism obtains. The worst case inefficiency of a Nash equilibrium in a model is captured by a concept called \emph{price of anarchy}. Price of anarchy is thus the ratio of the social welfare of the optimal direct revelation outcome to that of the worst-case Nash equilibrium.

\subsection{A Natural Modular Mechanism: VCG per Stage}\label{sec: vcg per stage}

We instantiate the general modular framework with a specific, natural choice for the two mechanisms. Specifically, we suggest running VCG for each stage separately.

Working backwards, in the second stage, arguably the most natural choice is VCG, which boils down to a second price auction on ``effective'' values. Concretely, after a suggestion $\ell$ has been chosen (in step (3)) and a signal $\sigma$ has been publicly observed (in step (4)), the platform holds the following auction: allocate the item to the advertiser with the highest \emph{effective} bid, i.e. the advertiser $i^*$ with the highest $b^A_i \cdot \alpha_i(\ell,\sigma)$, and charge them the second-highest effective bid $p^A_{i^*} = \max_{j \neq i^*} b_j^A \cdot \alpha_j(\ell,\sigma)$. 

Anticipating the truthful and efficient outcome of the second stage auction, each advertiser can calculate its expected utility for any given question $\ell$. Concretely, for a fixed question $\ell$, an advertiser can simulate a draw $\sigma$ from $S_{\ell}(.)$, and the resulting payoff from the subsequent second price auction. This induces a well-defined expected utility $R_i(\ell; v_i)$, for each advertiser $i$ and question $\ell$.
Therefore, the stage 1 problem boils down to picking a mechanism for the following ``public projects'' problem: there are $k$ projects to choose from (the questions), and $n$ agents (the advertisers), each with a private value $R_i(\ell; v_i)$ for each project. For this problem, VCG is again the natural, truthful solution: pick the project (question) that maximizes social welfare, and charge each agent its externality. Concretely, given a bid $b^Q_{i, \ell}$ from each advertiser $i$ for each question $\ell$: present the question $\ell^*$ with maximum sum of bids, i.e. $\ell^* = \operatorname*{argmax}_{\ell} \sum_{i=1}^n b^Q_{i, \ell}$, and charge each advertiser $i$ its externality $p^Q_i = \left( \max_{\ell} \sum_{j \neq i} b^Q_{j,\ell} \right) - \left( \sum_{j \neq i} b^Q_{j, \ell^*} \right)$.

\begin{example}
Consider the instance from \Cref{ex:running-shoes-formal}.

Assume that advertiser 1 submits:
\begin{itemize}[leftmargin=*]
\item Stage 1 bids: $b^Q_{1,\text{terr}}=21$, \quad $b^Q_{1,\text{tgt}} = 20$.
\item Stage 2 bidding function: $b_1^A(\ell,\sigma)=v_1=50$ for all $(\ell,\sigma)$.
\end{itemize}
Assume that advertiser 2 submits:
\begin{itemize}[leftmargin=*]
\item Stage 1 bids: $b^Q_{2,\text{terr}}=9$, \quad $b^Q_{2,\text{tgt}}=12$.
\item Stage 2 bidding function: $b_2^A(\ell,\sigma)=v_2=30$ for all $(\ell,\sigma)$.
\end{itemize}

Then the mechanism proceeds as follows. First, since $\sum_i b^Q_{i,\text{tgt}} = 20 + 12 > 21 + 9 = \sum_i b^Q_{i,\text{terr}}$, the question chosen in the first stage is $\ell^* = \ell_{\text{tgt}}$. To compute the stage 1 VCG externalities for advertiser $1$ we have:
\(\max_\ell b^Q_{2,\ell}=\max\{ 9 , 12\}= 12\) and
\( b^Q_{2,\ell^*}= 12\), so \(p_1^Q=0\). To compute the stage 1 VCG externalities for advertiser $2$ we have: \(\max_\ell b^Q_{1,\ell}=\max\{ 21, 20 \}= 21\) and
\( b^Q_{1,\ell^*}=20\), so \(p_2^Q = 1 \).

Then, a signal is realized in stage 2. 
Since $\ell_{\text{tgt}}$ was chosen:
\begin{itemize}[leftmargin=*]
\item If $\sigma=\text{click}$ (prob.\ $1/4$), the effective bids are
$(b_1^A\alpha_1,b_2^A\alpha_2)=(50\cdot 0.9,\ 30\cdot 0)=(45,0)$.
Advertiser~1 wins. The second-highest effective bid is $0$, so the
Stage~2 payment is \(p_1^A=0\).
\item If $\sigma=\text{no-click}$ (prob.\ $3/4$), the effective bids are
$(50\cdot 0.1,\ 30\cdot 0.4)=(5,12)$. Advertiser~2 wins and pays the second-highest effective bid \(5\), so \(p_2^A=5\).
\end{itemize}

\end{example}

\subsection{Results}

We first prove that the natural mechanism from~\Cref{sec: vcg per stage}, VCG per stage, has a natural pure Nash equilibrium: every advertiser bids their (true) expected utility for each question, as well as their true base value. We defer the proof of~\Cref{thm: pure Nash} to Appendix~\ref{app: missing}.

\begin{theorem}[Existence of a pure Nash equilibrium]\label{thm: pure Nash}
Consider the modular mechanism $\mathcal{M}=(\mathcal{M}^{VCG}_Q,\mathcal{M}^{VCG}_A)$ where $\mathcal{M}^{VCG}_A$ is the (per-signal) VCG auction on effective values and $\mathcal{M}^{VCG}_Q$ is VCG for choosing the question, as defined in~\Cref{sec: vcg per stage}. Define, for each advertiser $i$ and question $\ell$,
\[
R_i(\ell;v)\ :=\ E_{\sigma\sim S_\ell}\!\Big[\,u_i^{A}\big(v;\ell,\sigma\big)\,\Big],
\]
where $u_i^{A}(v;\ell,\sigma)$ is $i$’s (quasilinear) utility in $\mathcal{M}^{VCG}_A$ when the realized question is $\ell$, the realized signal is $\sigma$, and everyone bids their base value in stage 2. Then the strategy profile
\[
b_i^{A} (\ell,\sigma) = v_i\quad\text{for all }i, \ell, \sigma,\qquad
b_{i,\ell}^{Q}=R_i(\ell;v)\quad\text{for all }i,\ \ell
\]
is a pure Nash equilibrium.
\end{theorem}

Note that, as described in the first part of this section, advertisers \emph{simultaneously} submit (i) a bid for each question and (ii) a bidding function for the item, mapping $(\ell, \sigma)$ to a stage-2 bid. Suppose instead that the platform only collected base values $b_i$, and then computed stage-1 values $R_i(\ell;v)$ on the advertisers' behalf (using $v_i = b_i$). In this proxy variant, truthfully reporting $b_i = v_i$ is \emph{not} a Nash equilibrium. The proof of the following proposition is deferred to~\Cref{app: missing}.

\begin{proposition}\label{prop: proxy two stage is not truthful}
In the proxy variant, where the platform collects only base values $b_i$ and computes $R_i(\ell;b)$ to run $\mathcal{M}^{VCG}_Q$ (and uses $b_i$ in $\mathcal{M}^{VCG}_A$), reporting $b_i=v_i$ need not be a Nash equilibrium.
\end{proposition}

The equilibrium outcome in~\Cref{thm: pure Nash} might differ from the welfare-maximizing outcome (the outcome of the mechanism as described in~\Cref{sec: direct}). Specifically, while the stage 2 outcome of the equilibrium of~\Cref{thm: pure Nash} maximizes welfare for every realized question $\ell$ and signal $\sigma$, the stage 1 question selected maximizes expected \emph{quasilinear utility}
\[
\ell^*_{eq} = \operatorname*{argmax}_{\ell} \sum_{i=1}^n R_{i} (\ell; v) = \operatorname*{argmax}_{\ell} \sum_{i=1}^n E_{\sigma\sim S_\ell} \Big[\,u_i^{A}\big(v;\ell,\sigma\big)\,\Big]. \]
On the other hand, the welfare-maximizing outcome selects the question that maximizes expected \emph{value}
\[
\ell^*_{opt} = \operatorname*{argmax}_{\ell} E_{\sigma \sim S_{\ell}(.) } [ \max_{j=1,\dots,n} b_j \cdot \alpha_j(\ell, \sigma) ].
\]

Therefore, it is natural to ask how inefficient this equilibrium outcome can be. And, more generally, how bad is the Price of Anarchy of the proposed mechanism $\mathcal{M}=(\mathcal{M}^{VCG}_Q,\mathcal{M}^{VCG}_A)$?
Our next theorem shows that not only do inefficient equilibria exist, but even the pure Nash equilibrium suggested in~\Cref{thm: pure Nash} has infinite Price of Anarchy.

\begin{theorem}\label{thm: unbounded poa}
The Price of Anarchy of the mechanism $\mathcal{M}=(\mathcal{M}^{VCG}_Q,\mathcal{M}^{VCG}_A)$ is unbounded. Moreover, for all $c \geq 10$, there exists an instance such that the optimal welfare is at least $c$ times the welfare of the equilibrium from~\Cref{thm: pure Nash}.
\end{theorem}

\begin{proof}
Fix an integer $m\ge 3$ and a parameter $\delta > 0$ (fixed later in this proof).
Our instance is constructed as follows:
\begin{itemize}[leftmargin=*]
    \item States and questions:
    \begin{enumerate}
        \item $\Theta=\{1,2,\dots,m\}$ with the uniform prior $\D(\theta=t)=1/m$.
        \item A revealing question $\ell=1$: $\Sigma_1=\{\sigma_1,\dots,\sigma_m\}$ with $Q_1(\sigma_t\mid \theta=t)=1$. Thus $\D_{1,\sigma_t}$ is a point mass on state $t$.
        \item An uninformative question $\ell=2$: $\Sigma_2=\{\bar\sigma\}$ with $Q_2(\bar\sigma\mid \theta)=1$ for all $\theta$. Thus, $\D_{2,\bar\sigma}=\D$.
    \end{enumerate}
    \item Advertisers: 
        \begin{enumerate}
            \item There are $n=m$ advertisers, each with a base value $v_i=1$.
            \item For each state $i \in\{1,\dots,n\}$, $\alpha_{i}(i)=1$. That is, advertiser $i$ is the ``primary'' advertiser for state $i$.
            \item For every state $i \in \{1,\dots,n\}$, $\alpha_{i-1}(i) = 1-\delta$ (where $i-1 = n$ for $i=0$). That is, advertiser $i$ is the ``secondary'' advertiser for state $i+1$.
            \item For state $i=3$, $\alpha_1(3)=1-\delta$ and $\alpha_2(3) = 0$. That is, advertiser $1$ is the ``secondary'' advertiser for state 3 as well (also the secondary for state 2, and the primary for state 1), while advertiser $2$ is not the secondary advertiser for any state.
            \item All remaining conversion rates are $0$: $\alpha_i(t)=0$ for all other pairs $(i,t)$.
        \end{enumerate}
\end{itemize}

This construction ensures that we have exactly two advertisers with positive bids in each state (which creates competition under the revealing question). We now compute induced (signal) conversion rates and compare the two questions.

\paragraph{Induced conversion rates.} Under question $\ell=1$, for every realized $\sigma_t$ the posterior is a point mass at state $t$, therefore, $\alpha_i(1,\sigma_t)=\alpha_i(t)$ for every state $t$ (and signal $\sigma_t$) and advertiser $i$. On the other hand, under the uninformative question $\ell=2$, the posterior equals the prior; therefore, $\alpha_i(2,\bar\sigma)= E_{\theta\sim\D}\big[\alpha_i(\theta)\big] = \frac{1}{m}\sum_{t=1}^m \alpha_i(t)$.
From our construction, we have:
    \begin{align*}
      \alpha_1(2,\bar\sigma)&=\frac{1+2(1-\delta)}{m}=\frac{3-2\delta}{m} \\
      \alpha_2(2,\bar\sigma)&=\frac{1}{m} \\
      \alpha_i(2,\bar\sigma)&=\frac{1+(1-\delta)}{m}=\frac{2-\delta}{m}\quad\text{for all }i\notin\{1,2\}
    \end{align*}

\paragraph{Stage 2 outcome for $\mathcal{M}^{VCG}_A$}
Given any realized $(\ell,\sigma)$, $\mathcal{M}^{VCG}_A$ is a single-item VCG auction on effective values $b_i^A \, \alpha_i(\ell,\sigma) = v_i \, \alpha_i(\ell,\sigma) = \alpha_i(\ell,\sigma)$ (since truthful bidding is a dominant strategy, and $v_i=1$).

If $\ell=1$, in every state $t$, there are two positive effective values: $1$ from the primary advertiser (advertiser $t$) and $1-\delta$ from the secondary advertiser (advertiser $t-1$ for $t \notin \{ 1, 3 \}$, advertisers $1$ and  $2$ for $t=3$, and advertiser $n$ for $t=1$). Thus, for any state $t$, we have that the welfare (the value generated) is $1$, but the utility of the winning advertiser (the primary advertiser) is $\delta$.
Taking an expectation over the uniform prior we have that $R_i(1;v)=\delta/m$ for all $i$, where $R_i(\ell;v) = E_{\sigma\sim S_\ell}[u_i^{A}(v;\ell,\sigma)]$ is advertiser $i$'s utility under question $\ell$.

If $\ell=2$, there is a single signal $\bar\sigma$ with effective values
$\alpha_1(2,\bar\sigma) =\frac{3-2\delta}{m}$, $\alpha_2(2,\bar\sigma) =\frac{1}{m}$, and $\alpha_i(2,\bar\sigma) =\frac{2-\delta}{m}$ for all $i \notin \{ 1,2 \}$.
Therefore, under VCG (a second price auction) with these bids, we have that the (expected) value is $\frac{3-2\delta}{m}$ (for $\delta < 1$, we have $\frac{3-2\delta}{m} > \frac{2-\delta}{m}$), revenue $\frac{2-\delta}{m}$, and $R_i(2; v) = 0$ for $i \neq 1$, but $R_1(2;v) = \frac{3-2\delta}{m} - \frac{2-\delta}{m} = \frac{1-\delta}{m}$.

\paragraph{Stage 1 outcome for $\mathcal{M}^{VCG}_Q$} 
Recall that the equilibrium from~\Cref{thm: pure Nash} prescribes that $b_{i,\ell}^{Q}=R_i(\ell;v)$, for all advertisers $i$ and questions $\ell$. Given the stage 2 calculations, we therefore have that the total bid on question $\ell = 1$ is $\sum_{i=1}^n R_i(1;v)=\delta$. On the other hand, for question $\ell = 2$ we have that the total bid is $\sum_{i=1}^n R_i(2;v) = R_1(2;v) = \frac{1-\delta}{m}$.
By picking $\delta$ such that $\frac{1-\delta}{m} > \delta$ (e.g., $\delta = 1/m^2$), we have that the mechanism picks question $\ell = 2$.

Therefore, in the prescribed equilibrium for the mechanism $\mathcal{M}=(\mathcal{M}^{VCG}_Q,\mathcal{M}^{VCG}_A)$, the overall outcome is to pick question $\ell = 2$, which results in advertiser $1$ always winning in stage 2, for a total (expected) welfare of $\frac{3-2\delta}{m}$. On the other hand, presenting question $\ell = 1$ results in the ``primary'' advertiser to win the item in every state $t$, for a total (expected) welfare of $1$.
\end{proof}

It is tempting to conjecture that perhaps one can bypass~\Cref{thm: unbounded poa} by tweaking the mechanism. Arguably, the most natural choice would be to change the VCG payment in stage 1 to other natural payment formats, e.g., first-price payments, or all-pay payments (but still select the question with the highest total bid). As we show next, the result of~\Cref{thm: unbounded poa} is robust to such changes. The proof of~\Cref{thm: robust poa} is deferred to Appendix~\ref{app: missing}.

\begin{theorem}[Robustness of Unbounded PoA]\label{thm: robust poa}
The Price of Anarchy remains unbounded if the stage 1 VCG auction, $\mathcal{M}^{VCG}_Q$, is replaced by a mechanism that uses the same ``highest sum of bids'' allocation rule but with one of the following payment rules: (i) \textbf{First-Price:} Only the advertiser who submitted the single highest bid for the winning question pays that bid, and (ii) \textbf{All-Pay:} Every advertiser pays their submitted bid for the winning question.
\end{theorem}

\section{Conclusion}

In this work, we introduce a formal model for interactive sponsored search, capturing the growing trend of platforms eliciting user intent through suggestions before conducting a final ad auction. In this model, we analyzed two different design choices: an end-to-end, direct revelation mechanism that jointly optimizes the choice of suggestion and the subsequent ad allocation, and a simpler, modular two-stage mechanism that decouples the suggestion auction from the ad auction. Our results suggest that the direct revelation format provides overall better outcomes for the platform and advertisers.

Our work opens several avenues for future research. One direction is to move beyond welfare-maximization and study mechanisms designed for revenue maximization. Further afield, a natural evolution of this technology would be hybrid systems that present a mix of organic and sponsored follow-up questions, possibly over multiple rounds of interaction. Analyzing the complex dynamics of such hybrid systems, and how they affect the choice of mechanism, is an important next step.

\bibliographystyle{alpha}
\bibliography{refs}

\newpage

\appendix

\section{Missing proofs}\label{app: missing}

\begin{proof}[Proof of~\Cref{thm: pure Nash}]
For any realized $(\ell,\sigma)$, $\mathcal{M}^{VCG}_A$ is a single-item VCG auction on effective values $b_i^A\,\alpha_i(\ell,\sigma)$. Hence, $b_i^A=v_i$ is a dominant strategy for every advertiser $i$, independent of stage 1 behavior. Therefore, it suffices to argue that no stage 1 deviation exists, with stage 2 bids $b_i^A=v_i$ fixed.

Let $u_i^{A}(v;\ell,\sigma)$ denote $i$’s quasilinear utility in stage 2 outcome (fixing $b_i^A=v_i$ for all $i$), and let $R_i(\ell;v)\ :=\ E_{\sigma\sim S_\ell} \Big[\,u_i^{A}\big(v;\ell,\sigma\big)\,\Big]$ be each advertiser's ex-ante utility for implementing question $\ell$, fixing the behavior in stage 2.
With $b^A=v$ fixed, advertiser $i$'s total expected utility under stage 1 bids $b^Q$ is
\[
U_i(b^Q)\ =\ \sum_{\ell=1}^k \lambda_\ell(b^Q)\,R_i(\ell;v)\;-\;p_i^Q(b^Q).
\]

The prescribed strategy in stage 1 is $b_{i,\ell}^{Q} = R_i(\ell;v)$, for every advertiser $i$ and question $\ell$. 
Fix the bids of all other advertisers $j$ to $b_{j,\ell}^{Q} = R_j(\ell;v)$, and consider an alternative stage 1 bidding strategy $\hat{b}_{i,\ell}^{Q}$. Slightly overloading notation, it remains to argue that 
$U_i(R_i; R_{-i}) \geq U_i(\hat{b}_i^Q; R_{-i})$.

First, if the chosen question under $(\hat{b}_i^Q, R_{-i})$ is the same as the chosen question under $(R_i; R_{-i})$, then the stage 1 payment $p_i^Q(\hat{b}_i^Q, R_{-i})$ is equal to $p_i^Q(R_i, R_{-i})$ (since the externality is the same), so $U_i(\hat{b}_i^Q; R_{-i})=U_i(R_i; R_{-i})$. Second, consider the case that the chosen question $\hat{\ell}$ under $(\hat{b}_i^Q, R_{-i})$ is \emph{not} the same as the chosen question $\ell^*$ under $(R_i; R_{-i})$. Then we have
\begin{align*}
 U_i(\hat b_i^Q; R_{-i})
&= R_i(\hat\ell;v) - p_i^Q(\hat b_i^Q, R_{-i}) \\
&= \sum_{j=1}^n R_j(\hat\ell;v) - \max_{\ell}\sum_{j\neq i} R_j(\ell;v).   
\end{align*}
The second term does not depend on $i$. And, the first term, $\sum_{j=1}^n R_j(\hat{\ell} ; v)$ is maximized for $\hat{\ell} = \ell^*$. Therefore, 
\begin{align*}
U_i(\hat b_i^Q; R_{-i}) &= \sum_{j=1}^n R_j(\hat\ell;v) - \max_{\ell}\sum_{j\neq i} R_j(\ell;v) \\
&\leq \sum_{j=1}^n R_j(\ell^*;v) - \max_{\ell}\sum_{j\neq i} R_j(\ell;v) \\
&= U_i(R_i; R_{-i}).
\end{align*}
This concludes the proof that the prescribed strategy is a pure Nash equilibrium.
\end{proof}

\begin{proof}[Proof of~\Cref{prop: proxy two stage is not truthful}]
Consider an instance with two advertisers with base values $v_1 = v_2 = 10$. There are two questions, each producing two signals with equal probability.

Question A: With probability $1/2$ the effective conversion rates are $(\alpha_1,\alpha_2)=(0.6, 0)$, and with probability $1/2$ they are $(0, 0.8)$. Thus in each signal there is a unique winner and the second-highest effective bid is $0$.

Question B: with probability $1/2$ the rates are $(1,0.4)$; with probability $1/2$ they are $(1,0)$.

\paragraph{Truthful bidding: $b_i = v_i$} Under truthful bidding, Question A yields effective bids $(6, 0)$ and $(0, 8)$, each with probability $1/2$. Therefore
\[
R_1(A;v)=\frac{1}{2} \cdot 6 = 3,\qquad
R_2(A;v)=\frac{1}{2} \cdot 8 = 4.
\]
Question B yields effective bids $(10,4)$ (w.p. $1/2$) and $(10,0)$ (w.p. $1/2$). In the first signal, advertiser 1 wins and pays $4$; in the second, advertiser 1 wins and pays $0$. Hence
\[
R_1(B;v)= \frac{1}{2} \cdot 6 + \frac{1}{2} \cdot 10 = 8,\qquad
R_2(B;v)= 0.
\]
Stage 1 VCG is run under bids $\{R_i(\ell;v)\}$. We have $\sum_i R_i(A;v) = 7$ and $\sum_i R_i(B;v) = 8$, therefore Question B is picked. The utility of advertiser 2 is $U_2 = 0$.

\paragraph{A profitable deviation.} 
Consider the following deviation for advertiser 2: $b_2 = 20$ (while still $b_1 = v_1 = 10$).

We have that Question A yields effective bids $(6,0)$ and $(0,16)$, each with probability $1/2$. Therefore:
\[
R_1(A;b)= \frac{1}{2} \cdot 6 = 3,\qquad R_2(A;b)= \frac{1}{2} \cdot 16 = 8.
\]
Question B yields effective bids $(10,8)$ (w.p. $1/2$) and $(10,0)$ (w.p. $1/2$). In the first signal, advertiser 1 wins and pays $8$; in the second, advertiser 1 wins and pays $0$. Hence
\[
R_1(B;b)= \frac{1}{2} \cdot 2 + \frac{1}{2} \cdot 10 = 6,\qquad
R_2(B;b)= 0.
\]
Therefore, this time the question with the highest bids is Question A ($\sum_i R_i(A;b) = 11$ and $\sum_i R_i(B;b) = 6$). Advertiser 2 imposes an externality: $\max\{ R_1(A;b), R_1(B;b) \} - R_1(A;b) = 6-3 = 3$, so $p_2^Q = 3$. And, the \emph{true} value that advertiser 2 has for Question A is $R_2(A;v) = 4$ (calculated earlier, from the truthful bid). Therefore, advertiser 2's utility is $4 - p_2^Q = 1$, which is better than the utility of zero under truthful bidding.
\end{proof}

\begin{proof}[Proof of~\Cref{thm: robust poa}]
We use the instance from the proof of~\Cref{thm: unbounded poa}. We begin by noting that the second-stage mechanism is a VCG auction, which is dominant-strategy incentive compatible. Advertisers will therefore bid their true base values ($b_i^A = v_i$) in the second stage. This allows us to use the derived expected utilities, $R_i(\ell;v)$, as the advertisers' values in the first-stage auction:
\begin{itemize}[leftmargin=*]
    \item $R_i(1;v) = \delta/m$ for all advertisers $i$.
    \item $R_1(2;v) = \frac{1-\delta}{m}$, and $R_j(2;v) = 0$ for all $j \neq 1$.
\end{itemize}

Now, consider a first-price auction (the advertiser with the highest bid on the winning question pays this bid) or all-pay auction (all advertisers pay their bid on the winning question) for the first stage, where the question with the highest sum of bids wins. 
We will assume tie-breaking in favor of question 2: that is, if the sum of bids is the same, then question 2 (the uninformative question) is chosen. The proof can be easily adjusted to work under tie-breaking that favors the first question.

We claim the following strategy profile is a Nash Equilibrium. Advertiser 1 bids $b_{1,2}^Q = \delta$ for the uninformative question $\ell=2$ and $b_{1,1}^Q=0$. All other advertisers $j \neq 1$ bid $0$ for the uninformative question, and $\frac{\delta}{m}$ for the informative question.

In this profile, the inefficient question $\ell=2$ is chosen (the sum of bids on both questions is $\delta$; by our tie-breaking rule, question 2 wins), yielding a welfare of $\frac{3-2\delta}{m}$, while the optimal welfare is $1$. It remains to argue that this is an equilibrium.

First, advertiser 1 does not want to increase its bid for the winning question, since this action would increase its payment without changing the outcome. Increasing its bid for the losing question will change the outcome to question 1 being picked, which gives utility at most $\frac{\delta}{m}$, which is strictly less than her current utility of $R_1(2;v) - \delta = \frac{1-\delta}{m} - \delta$, for all $\delta < \frac{1}{m+2}$. Advertiser 1 cannot decrease its bid for the losing question (it's already zero). Finally, advertiser 1 does not want to decrease its bid for the winning question: any decrease of $\epsilon > 0$ would result in a different question being asked. Currently, advertiser $1$ gets utility $R_1(2;v) - \delta = \frac{1-\delta}{m} - \delta$ which is at least $R_1(1;v) = \frac{\delta}{m}$ for all $\delta < \frac{1}{m+2}$.

Second, the only thing that an advertiser $i \neq 1$ can do to affect the outcome is to bid strictly more than $\frac{\delta}{m}$ on question 1. However, $R_i(1;v)= \frac{\delta}{m}$, so, since under both payment rules (first-price and all-pay) they will end up paying their bid, this deviation is not profitable. Finally, if advertiser $i$'s deviation, for $i \neq 1$, does not affect the outcome (question asked), then its payment cannot be further reduced (since its bidding zero for the winning question already).
\end{proof}

\end{document}